\documentclass[sigconf]{acmart}

\title{FA*IR: A Fair Top-k Ranking Algorithm}



\setcopyright{none}             

\newcommand{\monthName}{\ifcase \month \or January\or February\or March\or %
	April\or May \or June\or July\or August\or September\or October\or November\or %
	December\fi}

\copyrightyear{2017}
\acmYear{2017}
\setcopyright{acmlicensed}
\acmConference{CIKM'17}{}{November 6--10, 2017, Singapore.}
\acmPrice{15.00}
\acmDOI{https://doi.org/10.1145/3132847.3132938}
\acmISBN{ISBN 978-1-4503-4918-5/17/11}

\usepackage{balance}  
\usepackage{graphics} 
\usepackage{paralist} 
\usepackage{booktabs} 
\usepackage{makecell,xcolor,amsmath}

\usepackage{xspace}
\usepackage{subfigure}
\usepackage[utf8]{inputenc}
\usepackage[T1]{fontenc}

\usepackage[ruled,lined,linesnumbered]{algorithm2e}
\SetKwInOut{AlgInput}{input}
\SetKwInOut{AlgOutput}{output}
\SetKwComment{AlgComment}{// }{}

\usepackage{hyperref}
\usepackage{xcolor}
\hypersetup{
	colorlinks=false,
	linkcolor={red!20!black},
	citecolor={green!20!black},
	urlcolor={blue!20!black}
}

\newcommand{\spara}[1]{\smallskip\noindent{\bf #1}}

\newcommand{\algoFAIR}[0]{{\textsc{FA*IR}}\xspace}
\newcommand{\algoFAIRBF}[0]{{\textsc{\textbf{FA*IR}}}\xspace}

\newcommand{\adj}[0]{\ensuremath{\operatorname{c}}}
\newcommand{\alphaadj}[0]{\ensuremath{\alpha_c}}

\newcommand{\algoCorrect}[0]{{\sc AdjustSignificance}\xspace}

\begin{document}


	\author{Meike Zehlike}
	\affiliation{%
		\institution{TU Berlin}
		\city{Berlin}
		\country{Germany}
	}
	\email{meike.zehlike@tu-berlin.de}

	\author{Francesco Bonchi}
	\affiliation{%
		\institution{ISI Foundation}
		\city{Turin}
		\country{Italy}
	}
	\email{francesco.bonchi@isi.it}

	\author{Carlos Castillo}
	\affiliation{%
		\institution{Universitat Pompeu Fabra}
		\city{Barcelona}
		\state{Catalunya}
		\country{Spain}
	}
	\email{chato@acm.org}

	\author{Sara Hajian}
	\affiliation{%
		\institution{NTENT}
		\city{Barcelona}
		\state{Catalunya}
		\country{Spain}
	}
	\email{shajian@ntent.com}

	\author{Mohamed Megahed}
	\affiliation{%
		\institution{TU Berlin}
		\city{Berlin}
		\country{Germany}
	}
	\email{mohamed.megahed@campus.tu-berlin.de}

	\author{Ricardo Baeza-Yates}
	\affiliation{%
		\institution{Universitat Pompeu Fabra}
		\city{Barcelona}
		\state{Catalunya}
		\country{Spain}
	}
	\email{rbaeza@acm.org}

	\renewcommand{\shortauthors}{M. Zehlike et al.}


\begin{abstract}
	In this work, we define and solve the Fair Top-$k$ Ranking problem, in which we want to determine a subset of $k$ candidates from a large pool of $n \gg k$ candidates, maximizing utility (i.e., select the ``best'' candidates) subject to group fairness criteria.

	Our ranked group fairness definition extends group fairness using the standard notion of protected groups and is based on ensuring that the proportion of protected candidates in every prefix of the top-$k$ ranking remains statistically above or indistinguishable from a given minimum.
	Utility is operationalized in two ways:
	\begin{inparaenum}[(i)]
		\item every candidate included in the top-$k$ should be more qualified than every candidate not included; and
		\item for every pair of candidates in the top-$k$, the more qualified candidate should be ranked above.
	\end{inparaenum}

	An efficient algorithm is presented for producing the Fair Top-$k$ Ranking, and tested experimentally on existing datasets as well as new datasets released with this paper, showing that our approach yields small distortions with respect to rankings that maximize utility without considering fairness criteria.
	To the best of our knowledge, this is the first algorithm grounded in statistical tests that can mitigate biases in the representation of an under-represented group along a ranked list. 
\end{abstract}
\maketitle

\vspace*{-\baselineskip}
\section{Introduction}\label{sec:introduction}
People search engines are increasingly common for job recruiting and even for finding companionship or friendship.
A top-$k$ ranking algorithm is typically used to find the most suitable way of ordering items (persons, in this case), considering that if the number of people matching a query is large, most users will not scan the entire list.
Conventionally, these lists are ranked in descending order of some measure of the relative quality of items.

The main concern motivating this paper is that a biased machine learning model that produces ranked lists can further systematically reduce the visibility of an already disadvantaged group~\cite{peder2008,Dwork2012} (corresponding to a legally protected category such as people with disabilities, racial or ethnic minorities, or an under-represented gender in a specific industry).

According to \cite{friedman1996bias} a computer system is \emph{biased} ``if it systematically and unfairly discriminate[s] against certain individuals or groups of individuals in favor of others. A system discriminates unfairly if it denies an opportunity or a good or if it assigns an undesirable outcome to an individual or a group of individuals on grounds that are unreasonable or inappropriate.''
Yet ``unfair discrimination alone does not give rise to bias unless it occurs systematically'' and ``systematic discrimination does not establish bias unless it is joined with an unfair outcome.''
On a ranking, the desired good for an individual is to appear in the result and to be ranked amongst the top-$k$ positions. The outcome is unfair if members of a protected group are systematically ranked lower than those of a privileged group.
The ranking algorithm discriminates unfairly if this ranking decision is based fully or partially on the protected feature. This discrimination is systematic when it is embodied in the algorithm's ranking model. As shown in earlier research, a machine learning model trained on datasets incorporating \textit{preexisting bias} will embody this bias and therefore produce biased results, potentially increasing any disadvantage further, reinforcing existing bias~\cite{oneil2016weapons}.

Based on this observation, in this paper we study the problem of producing a fair ranking given one legally-protected attribute,\footnote{We make the simplifying assumption that there is a dominant legally-protected attribute of interest in each case. The extension to deal with multiple protected attributes is left for future work.} i.e., a ranking in which the representation of the minority group does not fall below a minimum proportion $p$ \emph{at any point in the ranking}, while the utility of the ranking is maintained as high as possible.

We propose a post-processing method to remove the systematic bias by means of a \emph{ranked group fairness criterion}, that we introduce in this paper. We assume a ranking algorithm has given an undesirable outcome to a group of individuals, but the algorithm itself cannot determine if the grounds were appropriate or not. Hence we expect the user of our method to know that the outcome is based on unreasonable or inappropriate grounds and provide $p$ as input which can originate in a legal mandate or in voluntary commitments.
For instance, the US Equal Employment Opportunity Commission sets a goal of 12\% of workers with disabilities in federal agencies in the US,\footnote{US EEOC: \url{https://www1.eeoc.gov/eeoc/newsroom/release/1-3-17.cfm}, Jan 2017.}
while in Spain, a minimum of 40\% of political candidates in voting districts exceeding a certain size must be women~\cite{verge2010gendering}.
In other cases, such quotas might be adopted voluntarily, for instance through a diversity charter.\footnote{European Commission: \url{http://ec.europa.eu/justice/discrimination/diversity/charters/}}
In general these measures do not mandate perfect parity, as distributions of qualifications across groups can be unbalanced for legitimate, explainable reasons~\cite{vzliobaite2011handling,pedreschi2009integrating}. 

The ranked group fairness criterion compares the number of protected elements in every prefix of the ranking with the expected number of protected elements if they were picked at random using Bernoulli trials (independent ``coin tosses'') with success probability $p$.
Given that we use a statistical test for this comparison, we include a significance parameter $\alpha$ corresponding to the probability of a Type I error, which means rejecting a fair ranking in this test.

\medskip
\textit{Example.} Consider the three rankings in Table \ref{tbl:xing_intro_example} corresponding to searches for an ``economist,'' ``market research analyst,'' and ``copywriter'' in XING\footnote{\url{https://www.xing.com/}}, an online platform for jobs that is used by recruiters and headhunters, mostly in German-speaking countries, to find suitable candidates in diverse fields (this data collection is reported in detail on \S\ref{concept:XING}). While analyzing  the extent to which candidates of both genders are represented as we go down these lists,  we can observe that such proportion keep changing and is not uniform (see, for instance, the top-10 vs. the top-40). As a consequence, recruiters examining these lists will see different proportions depending on the point at which they decide to stop.
Corresponding with \cite{friedman1996bias} this outcome systematically disadvantages individuals of one gender by preferring the other at the top-$k$ positions. As we do not know the learning model behind the ranking, we assume that the result is at least partly based on the protected attribute \emph{gender}.

Let $k = 10$. Our notion of \textit{ranked group fairness} imposes a fair representation with proportion $p$ and significance $\alpha$ at each top-$i$ position with $i \in [1,10]$ (formal definitions are given in \S\ref{sec:problem}).
Consider for instance $\alpha = 0.1$ and suppose that the required proportion is $p = 0.4$.  This translates (see Table \ref{tbl:ranked_group_fairness_table}) to having at least one individual from the protected minority class in the first 5 positions: therefore the ranking for  ``copywriter'' would be rejected as unfair. However, it also requires to have at least 2 individuals from the protected group in the first 9 positions: therefore also the ranking for ``economist'' is rejected as unfair, while the ranking for ``market research analyst'' is fair for  $p = 0.4$. However, if we would require $p = 0.5$ then this translates in having at least 3 individuals from the minority group in the top-10, and thus even the ranking for ``market research analyst'' would be considered unfair.
We note that for simplicity, in this example we have not adjusted the significance $\alpha$ to account for multiple statistical tests; this is not trivial, and is one of the key contributions of this paper.
\medskip

\begin{table}[t]
	\caption{Example of non-uniformity of the top-10 vs. the top-40 results for different queries in XING (Jan 2017).
		\label{tbl:xing_intro_example}}

	\vspace{-3mm}

	\centering\small\begin{tabular}{lccccc}\toprule
		& Position					  & top 10 & top 10  & top 40 & top 40 \\
		& \texttt{1 2 3 4 5 6 7 8 9 10} & male & female & male & female \\
		\midrule
		Econ.  & \texttt{f m m m m m m m m m} & 90\% & 10\% & 73\% & 27\% \\
		Analyst& \texttt{f m f f f f f m f f} & 20\% & 80\% & 43\% & 57\% \\
		Copywr.& \texttt{m m m m m m f m m m} & 90\% & 10\% & 73\% & 27\% \\
		\bottomrule
	\end{tabular}

	\vspace{-2mm}
\end{table}

\spara{Our contributions.}
In this paper, we define and analyze the {\sc Fair Top-$k$ Ranking problem}, in which we want to determine a subset of $k$ candidates from a large pool of $n \gg k$ candidates, in a way that maximizes utility (selects the ``best'' candidates), subject to group fairness criteria. The running example we use in this paper is that of selecting automatically, from a large pool of potential candidates, a smaller group that will be interviewed for a position.

Our notion of utility assumes that we want the interview the most qualified candidates, while their qualification is equal to a relevance score calculated by a ranking algorithm.
This score is assumed to be based on relevant metrics for evaluating candidates for a position, which depending on the specific skills required for the job could be their grades (e.g., Grade Point Average), their results in a standardized knowledge/skills test specific for a job, their typing speed in words per minute for typists, or their number of hours of flight in the case of pilots.
We note that this measurement will embody \emph{preexisting bias} (e.g. if black pilots are given less opportunities to flight they accumulate less flight hours), as well as \emph{technical bias}, as learning algorithms are known to be susceptible to direct and indirect discrimination~\cite{tuto2016,HajianFerrer12}.

The utility objective is operationalized in two ways.
First, by a criterion we call \emph{selection utility}, which prefers rankings in which every candidate included in the top-$k$ is more qualified than every candidate not included, or in which the difference in their qualifications is small.
Second, by a criterion we call \emph{ordering utility}, which prefers rankings in which for every pair of candidates included in the top-$k$, either the more qualified candidate is ranked above, or the difference in their qualifications is small.

Our definition of \emph{ranked group fairness} reflects the legal principle of group under-representation in obtaining a benefit \cite{ellis2012eu,lerner2003group}. We use the standard notion of a protected group (e.g., ``people with disabilities''); where protection emanates from a legal mandate or a voluntary commitment.
%
%
We formulate a criterion applying a statistical test on the proportion of protected candidates on every prefix of the ranking, which should be indistinguishable or above a given minimum.
%
%
We also show that the verification of the ranked group fairness criterion can be implemented efficiently.

Finally, we propose an efficient algorithm, named \algoFAIR, for producing a top-$k$ ranking that maximizes utility while satisfying ranked group fairness, as long as there are ``enough'' protected candidates to achieve the desired minimum proportion.
We also present extensive experiments using both existing and new datasets to evaluate the performance of our approach compared to the so-called ``color-blind'' ranking with respect to both the utility of ranking and the fairness degree.

Summarizing, the main contributions of this paper are:
\begin{compactenum}
	\item the principled definition of \emph{ranked group fairness}, and the associated  {\sc Fair Top-$k$ Ranking problem};
	\item the \algoFAIR\ algorithm for producing a top-$k$ ranking that maximizes utility while satisfying ranked group fairness.
\end{compactenum}

Our method can be used within an anti-discrimination framework such as \emph{positive actions}~\cite{sowell2005affirmative}. We do not claim these are the only way of achieving fairness, but we provide \emph{an algorithm grounded in statistical tests that enables the implementation of a positive action policy in the context of ranking}.

The rest of this paper is organized as follows.
The next section presents a brief survey of related literature, while Section~\ref{sec:problem} introduces our ranked group fairness and utility criteria, our model adjustment approach, and a formal problem statement.
Section~\ref{sec:algorithms} describes the \algoFAIR\ algorithm.
Section~\ref{sec:experiments} presents experimental results.
Section~\ref{sec:conclusions} presents our conclusions and future work.

\section{Related Work}\label{sec:related-work}

Anti-discrimination has only recently been considered from an algorithmic perspective~\cite{tuto2016}. Some proposals are oriented to discovering and measuring discrimination (e.g., \cite{peder2008,Bonchi2015,angwin_2016_machine}); while others deal with mitigating or removing discrimination 
(e.g., \cite{CaldersICDM,HajianFerrer12,hajian2014,Dwork2012,Zemel2013}).
All these methods are known as \emph{fairness-aware algorithms}.

\subsection{Group fairness and individual fairness}
Two basic frameworks have been adopted in recent studies on algorithmic discrimination: \begin{inparaenum}[(i)]
	\item \emph{individual fairness}, a requirement that individuals should be treated consistently~\cite{Dwork2012}; and
	\item \emph{group fairness}, also known as statistical parity, a requirement that the protected groups should be treated similarly to the advantaged group or the populations as a whole \cite{peder2008,pederruggi2009}.
\end{inparaenum}

Different fairness-aware algorithms have been proposed to achieve group and/or individual fairness, mostly for predictive tasks. \citet{Calders2010} consider three approaches to deal with naive Bayes models by modifying the learning algorithm.
\citet{CaldersICDM} modify the entropy-based splitting criterion in decision tree induction to account for attributes denoting protected groups.
\citet{Kamishima2012}  apply a regularization (i.e., a
change in the objective minimization function) to probabilistic discriminative models, such
as logistic regression. \citet{zafar2015} describe fairness constraints for several classification methods.

\citet{Feldman2015} study \emph{disparate impact} in data, which corresponds to an unintended form of group discrimination, in which a protected group is less likely to receive a benefit than a non-protected group~\cite{Barocas2014}.
Besides measuring disparate impact, the authors also propose a method for removing it: we use this method as one of our experimental baselines in \S\ref{sec:experiments-baselines}. Their method ``repairs'' the scores of the protected group to make them have the same or similar distribution as the scores of the non-protected group, which is one particular form of positive action.
%
%
Recently, other fairness-aware algorithms have been proposed for mostly supervised learning algorithms and different bias mitigation strategies \cite{hardt2016equality, jabbari2016fair, friedler2016possibility, celis2016fair, corbett2017algorithmic}.


\subsection{Fair Ranking}
%
%
%
\citet{yang2016measuring} studied the problem of fairness in rankings.
They propose a statistical parity measure based on comparing the distributions of protected and non-protected candidates (for instance, using KL-divergence) on different prefixes of the list (e.g., top-10, top-20, top-30) and then averaging these differences in a discounted manner.
The discount used is logarithmic, similarly to Normalized Discounted Cumulative Gain (NDCG, a popular measure used in Information Retrieval~\cite{jarvelin2002cumulated}).
Finally, they show very preliminary results on incorporating their statistical parity measure into an optimization framework for improving fairness of ranked outputs while maintaining accuracy.
We use the synthetic ranking generation procedure of~\citet{yang2016measuring} to calibrate our method, and optimize directly the utility of a ranking that has statistical properties (ranked group fairness) resembling the ones of a ranking generated using that procedure; in other words, unlike~\cite{yang2016measuring}, we connect the creation of the ranking with the metric used for assessing fairness. 


\citet{kulshrestha_2017_quantifying}  determine search bias in rankings, proposing a quantification framework that measures the bias of the results of a search engine.
This framework discerns to what extent this output bias is due to the input dataset that feeds into the ranking system, and how much is due to the bias introduced by the system itself. 
%
In contrast to their work, which mostly focus on auditing ranking algorithms to identify the sources of bias in the data or algorithm, our paper focuses on generating fair ranked results.

A recent work~\cite{celis2017ranking} proposes algorithms for constrained ranking, in which the constraint is a $k \times \ell$ matrix with $k$ the length of the ranking and $\ell$ the number of classes, indicating the maximum number of elements of each class (protected or non-protected in the binary case) that can appear at any given position in the ranking. The objective is to maximize a general utility function that has a positional discount, i.e., gives more weight to placing a candidate with more qualifications in a higher position.
Differently from~\cite{celis2017ranking}, in our work we show how to construct the constraint matrix by means of a statistical test of ranked group fairness (a problem they left open), and our measure of utility is based on individuals, which allows to determine which individuals are more affected by the re-ranking with respect to the non-fairness-aware solution \S\ref{concept:our-utility-individual-fairness}.
\vspace*{-0.1cm}
\subsection{Diversity}

Additionally, the idea that we want to avoid showing only items of the same class has been studied in the Information Retrieval community for many years. The motivation there is that the user query may have different intents and we want to cover several of them with results. The most common approach, since \citet{carbonell1998use}, is to consider distances between elements, and maximize a combination of relevance (utility) with a penalty for adding to the ranking an element that is too similar to an element already appearing at a higher position.
A similar idea is used for diversification in recommender systems through various methods~\cite{kunaver2017diversity,channamsetty2017recommender}. They deal with different kinds of bias such as presentation bias, where, only a few items are shown and most of the items are not shown, and also popularity bias and a negative bias towards new items.
An exception is \citet{sakai2011evaluating}, that provides a framework for per-intent NDCG for evaluating diversity, in which an ``intent'' could be mapped to a protected/non-protected group in the fairness ranking setting. Their method, however, is concerned with evaluating a ranking, similar to the NDCG-based metrics described by~\citet{yang2016measuring} that we describe before, and not with a construction of such ranking, as we do in this paper. In contrast with most of the research on diversity of ranking results or recommender systems, our work operates on a discrete set of classes (not based on similarity to previous items).



\section{The Fair Top-k Ranking Problem}\label{sec:problem}

In this section, we first present the needed notation (\S\ref{subsec:preliminaries}), then the ranked group fairness criterion (\S\ref{subsec:group-fairness}-\S\ref{subsec:group-fairness-correction}) and criteria for utility (\S\ref{subsec:individual-fairness}). Finally we provide a formal problem statement (\S\ref{subsec:problem-statement}).


\subsection{Preliminaries and Notation}
\label{subsec:preliminaries}
\spara{Notation.}
Let $[n] = \{ 1, 2, \dots, n \}$ represent a set of candidates; let $q_i$ for $i \in [n]$ denote the ``quality'' of candidate $i$: this can be interpreted as an overall summary of the fitness of candidate $i$ for the specific job, task, or search query, and that could be obtained by the combination of several different attributes, possibly by means of a machine learning model, and potentially including preexisting and technical bias with respect to the protected group.
%
%
We will consider two kinds of candidates, protected and non-protected, and we will assume there are enough of them, i.e., at least $k$ of each kind. Let $g_i = 1$ if candidate $i$ is in the protected group, $g_i = 0$ otherwise.
Let ${\mathcal P}_{k,n}$ represent all the subsets of $[n]$ containing exactly $k$ elements.
Let ${\mathcal T}_{k,n}$ represent the union of all permutations of sets in ${\mathcal P}_{k,n}$.
For a permutation $\tau \in {\mathcal T}_{k,n}$ and an element $i \in [n]$, let
\[
r(i, \tau) = \begin{cases}
\mathrm{rank~of~} i \mathrm{~in~} \tau & \mathrm{if~} i \in \tau~, \\
|\tau| + 1 & \mathrm{otherwise}.
\end{cases}
\]

We further define $\tau_p$ to be the number of protected elements in $\tau$, i.e., $\tau_p = | \{ i \in \tau: g_i = 1 \} |$.
Let $c \in {\mathcal T}_{n,n}$ be a permutation such that $\forall i,j \in [n], r(i,c) < r(j,c) \Rightarrow q_i \ge q_j$. We call this the \emph{color-blind} ranking of $[n]$, because it ignores whether elements are protected or non-protected.
Let $c|_k = \langle c(1), c(2), \dots, c(k) \rangle$ be a prefix of size $k$ of this ranking. \label{concept:color-blind-ranking}
%


\spara{Fair top-$k$ ranking criteria.}\label{concept:criteria}
We would like to obtain $\tau \in {\mathcal T}_{k,n}$ with the following characteristics, which we describe formally next: 

\begin{enumerate}[{Criterion} 1.]
	\item Ranked group fairness: $\tau$ should fairly represent the protected group; \label{cond:ranking}

	\item Selection utility: $\tau$ should contain the most qualified candidates; and \label{cond:selection}

	\item Ordering utility: $\tau$ should be ordered by decreasing qualifications.\label{cond:ordering}
\end{enumerate}

We will provide a formal problem statement in \S\ref{subsec:problem-statement}, but first, we need to provide a formal definition of each of the criteria. 

\subsection{Group Fairness for Rankings}
\label{subsec:group-fairness}

We operationalize criterion~\ref{cond:ranking} of \S\ref{concept:criteria} by means of a \emph{ranked group fairness criterion}, which takes as input a protected group and a minimum target proportion of protected elements in the ranking, $p$. Intuitively, this criterion declares the ranking as unfair if the observed proportion is far below the target one.

Specifically, the ranked group fairness criterion compares the number of protected elements in every prefix of the ranking, with the expected number of protected elements if they were picked at random using Bernoulli trials.
The criterion is based on a statistical test, and we include a significance parameter ($\alpha$) corresponding to the probability of rejecting a fair ranking (i.e., a Type I error).
%

\begin{definition}[Fair representation condition]
	Let $F(x;n,p)$ be the cumulative distribution function for a binomial distribution of parameters $n$ and $p$. 
	A set $\tau \subseteq \mathcal{T}_{k,n}$, having $\tau_p$ protected candidates fairly represents the protected group with minimal proportion $p$ and significance $\alpha$,
	if $F(\tau_p;k,p) > \alpha$.
\end{definition}

This is equivalent to using a statistical test where the null hypothesis $H_0$ is that the protected elements are represented with a sufficient proportion $p_t$ ($p_t \ge p$), and the alternative hypothesis $H_a$

\noindent is that the proportion of protected elements is insufficient ($p_t < p$). In this test, the p-value is $F(\tau_p; k, p)$ and we reject the null hypothesis, and thus declare the ranking as unfair, if the p-value is less than or equal to the threshold $\alpha$.

The ranked group fairness criterion enforces the fair representation constraint over all prefixes of the ranking:

\begin{definition}[Ranked group fairness condition]
	A ranking $\tau \in {\mathcal T}_{k,n}$ satisfies the ranked group fairness condition with parameters $p$ and $\alpha$, if for every prefix $\tau|_i = \langle \tau(1), \tau(2), \dots, \tau(i) \rangle$ with $1 \le i \le k$, the set $\tau|_i$ satisfies the fair representation condition with proportion $p$ and significance $\alphaadj = \adj(\alpha, k, p)$.
	Function $\adj(\alpha, k, p)$ is a corrected significance to account for multiple testing (described in \S\ref{subsec:group-fairness-correction}).
\end{definition}

We remark that a larger $\alpha$ means a larger probability of declaring a fair ranking as unfair. In our experiments (\S\ref{sec:experiments}), we use a relatively conservative setting of $\alpha=0.1$.
The ranked group fairness condition can be used to create a \emph{ranked group fairness measure}. For a ranking $\tau$ and probability $p$, the ranked group fairness measure is the maximum $\alpha \in [0,1]$ for which $\tau$ satisfies the ranked group fairness condition.
Larger values indicate a stricter adherence to the required number of protected elements at each position.

\spara{Verifying ranked group fairness.}
Note that ranked group fairness can be verified efficiently in time $O(k)$, by having a pre-computed table of the percent point function with parameters $k$ and $p$, i.e, the inverse of $F(x;k,p)$.
Table~\ref{tbl:ranked_group_fairness_table} shows an example of such a table, computed for $\alpha=0.1$.
For instance, for $p=0.5$ we see that at least 1 candidate from the protected group is needed in the top 4 positions, and 2 protected candidates in the top 7 positions.

\begin{table}[t!]
	\caption{Example values of $m_{\alpha,p}(k)$, the minimum number of candidates in the protected group that must appear in the top $k$ positions to pass the ranked group fairness criteria with $\alpha=0.1$.}
	\vspace{-3mm}
	\label{tbl:ranked_group_fairness_table}
	\small\begin{tabular}{r|cccccccccccc}
		\diaghead{some text}%
		{p}{k}&
		1 & 2 & 3 & 4 & 5 & 6 & 7 & 8 & 9 & 10 & 11 & 12 \\ \midrule
		0.1      & 0 & 0 & 0 & 0 & 0 & 0 & 0 & 0 & 0 & 0  &  0 &  0 \\
		0.2      & 0 & 0 & 0 & 0 & 0 & 0 & 0 & 0 & 0 & 0  &  1 &  1 \\
		0.3      & 0 & 0 & 0 & 0 & 0 & 0 & 1 & 1 & 1 & 1  &  1 &  2 \\
		0.4      & 0 & 0 & 0 & 0 & 1 & 1 & 1 & 1 & 2 & 2  &  2 &  3 \\
		0.5      & 0 & 0 & 0 & 1 & 1 & 1 & 2 & 2 & 3 & 3  &  3 &  4 \\
		0.6      & 0 & 0 & 1 & 1 & 2 & 2 & 3 & 3 & 4 & 4  &  5 &  5 \\
		0.7      & 0 & 1 & 1 & 2 & 2 & 3 & 3 & 4 & 5 & 5  &  6 &  6 \\
		\bottomrule
	\end{tabular}
\end{table}

\subsection{Model Adjustment}
\label{subsec:group-fairness-correction}

Our ranked group fairness definition requires an adjusted significance $\alphaadj = \adj(\alpha, k, p)$.
This is required because it tests multiple hypotheses ($k$ of them). If we use $\alphaadj = \alpha$, we might produce false negatives, rejecting fair rankings, at a rate larger than $\alpha$.

The adjustment we propose is calibrated using the generative model of \citet{yang2016measuring}, which creates a ranking that we will consider fair by: \begin{inparaenum}[(i)]
	\item starting with an empty list, and
	\item incrementally adding the best available protected candidate with probability $p$, or the best available non-protected candidate with probability $1-p$.
\end{inparaenum}

\begin{figure}[t!]
	\centering\includegraphics[width=.55\columnwidth]{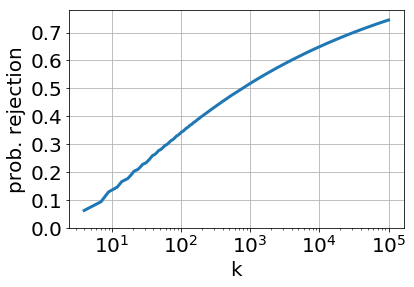}
	\vspace{-2mm}
	\caption{Example showing the need for multiple tests correction. The probability that a ranking generated by the method of \citet{yang2016measuring} with $p=0.5$ fails the ranked group fairness test with $p=0.5$ using $\alphaadj=0.1$, is in general larger than $\alpha=0.1$. Note the scale of $k$ is logarithmic.}
	\vspace{-3mm}
	\label{fig:example_plot_adjustment}
\end{figure}

Figure~\ref{fig:example_plot_adjustment} shows the probability that a fair ranking generated with $p=0.5$ is rejected by our ranked group fairness test with $p=0.5, ~\alphaadj=0.1$. The curve is computed analytically by the method we describe in the following paragraphs, and it experimentally matches the result of simulations we performed. 
We can see that the probability of a Type-I error (declaring this fair ranking as unfair) is in general higher than $\alpha = 0.1$.
If the $k$ tests were independent, we could use $\alphaadj = 1 - (1 - \alpha)^{1/k}$ (i.e., {\v S}id{\'a}k's correction), but given the positive dependence, the false negative rate is smaller than the bound given by {\v S}id{\'a}k's correction.
%

\newcommand{\mfail}{\ensuremath{\operatorname{fail}}}
\newcommand{\msucc}{\ensuremath{\operatorname{succ}}}
\newcommand{\minv}{m^{-1}}

The probability that a ranking generated using the process of \citet{yang2016measuring} with parameter $p$ passes the ranked group fairness criteria where each test is done with parameters $(p,\alphaadj)$ can be computed as follows:
Let $m(k) = m_{\alpha,p}(k)$ be as before the number of protected elements required up to position $k$.
Let $\minv(i) = k$ s.t. $m(k) = i$ be the position at which $i$ protected elements are required.
Let $b(i) = \minv(i) - \minv(i-1)$ (with $\minv(0) = 0$) be the size of a ``block,'' that is, the gap between one increase and the next in $m(\cdot)$.
An example is shown on Table~\ref{tbl:example_mtable}.
\begin{table}[t!]
	\caption{Example of $m(\cdot)$, $\minv(\cdot)$, and $b(\cdot)$ for $p=0.5, \alpha=0.1$.}
	\vspace{-3mm}
	\label{tbl:example_mtable}
	{\small
		\begin{tabular}{cccccccccccccc}\toprule
			$k$    & 1 & 2 & 3 & {\bf 4} & 5 & 6 & {\bf 7} & 8 & {\bf 9} & 10 & 11 & {\bf 12} \\
			\midrule
			$m(k)$ & 0 & 0 & 0 & \multicolumn{1}{c|}{1} & 1 & 1 & \multicolumn{1}{c|}{2} & 2 & \multicolumn{1}{c|}{3} & 3  & 3  & \multicolumn{1}{c}{4}\\
			Inverse   & \multicolumn{4}{c|}{$\minv(1)=4$}
			& \multicolumn{3}{c|}{$\minv(2)=7$}
			& \multicolumn{2}{c|}{$\minv(3)=9$}
			& \multicolumn{3}{c}{$\minv(4)=12$}\\
			Blocks       & \multicolumn{4}{c|}{$b(1)=4$}
			& \multicolumn{3}{c|}{$b(2)=3$}
			& \multicolumn{2}{c|}{$b(3)=2$}
			& \multicolumn{3}{c}{$b(4)=3$}\\\bottomrule
		\end{tabular}
	}
\end{table}

Let $I_\ell = \{ (i_1, i_2, ..., i_\ell): \forall \ell' \in [\ell], 0 \le i_{\ell'} \le b(\ell') \wedge \sum_{j=1}^{\ell'} i_j \ge \ell' \}$ represent all possible ways in which a fair ranking generated by the method of \citet{yang2016measuring} can pass the ranked group fairness test up to block $\ell$, with $i_j$ corresponding to the number of protected elements in block $1 \le j \le k$.
The probability of considering this ranking of $k$ elements ($\minv(k)$ blocks) unfair, is:
\begin{equation} \label{eq:failureProb}
1 - \sum_{v \in I_{\minv(k)}} \prod_{j=1}^{\minv(k)} f(v_j; b(j), p)
\end{equation}

\noindent where $f(x;b(j),p) = Pr(X = x)$ is the probability density function of a binomially-distributed variable $X \sim Bin(b(j), p)$.

\begin{algorithm}[t]
	\caption{Algorithm \algoCorrect used to compute model adjustment. Note that for notational convenience, vector indexes start at zero. Operator ``$>>$'' shifts vector components to the right, padding on the left with zeros.}
	\label{alg:correction} 
	\small
	\AlgInput{$k$, the size of the ranking to produce; $p$, the expected proportion of protected elements; $\alphaadj$, the significance for each individual test.}
	\AlgOutput{The probability of rejecting a fair ranking.}
	$(m_{\operatorname{old}},i_{\operatorname{old}}) \leftarrow (0, 0)$ \AlgComment{Auxiliary vectors}
	\For{$i \leftarrow 1$ \KwTo $k$}{
		$m[i] \leftarrow F^{-1}(\alphaadj; i, p)$ \\
		\If{$m[i] > m_{\operatorname{old}}$}{
			$\minv[m_{\operatorname{old}}] \leftarrow i$ \\
			$b[m_{\operatorname{old}}] \leftarrow i - i_{\operatorname{old}}$ \\
			$(m_{\operatorname{old}},i_{\operatorname{old}}) \leftarrow (m[i], i)$ \\
		}
	}
	$S[0] \leftarrow 1$ \AlgComment{Success probabilities}
	\For{$j \leftarrow 1$ \KwTo $m(k)$}{
		$S_{\operatorname{new}} \leftarrow$ zero vector of dimension $j$ \\
		\For{$i \leftarrow 0$ \KwTo $b(j)$}{
			\AlgComment{$f(i;b(j),p)$ is the prob. mass of $Bin(b(j),p)$}
			$S_{\operatorname{new}} \leftarrow S_{\operatorname{new}} + ( S >> i ) \cdot f(i; b(j), p)$ \\
		}
		$S_{\operatorname{new}}[j-1] \leftarrow 0$ \\
		$S \leftarrow S_{\operatorname{new}}$ \\
	}
	\Return{probability of rejecting a fair ranking: $1 - \sum S[i]$ }
\end{algorithm}

The above expression is intractable because of the large number of combinations in $I_{\minv(k)}$; however, there is an efficient iterative process to compute this quantity, shown in Algorithm~\ref{alg:correction}.
This algorithm maintains a vector $S$ that at iteration $\ell$ holds in position $S[i]$ the probability of having obtained $i$ protected elements in the first $\ell$ blocks, conditioned on obtaining at least $j$ protected elements up to each block $1 \le j \le \ell$. This has running time $O(k^2)$, but we note it has to be executed only once, as it does not depend on the dataset, only on $k$.
The summation of the probabilities in this vector $S$ is the probability that a fair ranking is accepted when using $\alphaadj$.
This algorithm can be used to determine the value of $\alphaadj$ at which the acceptance probability becomes $1-\alpha$, for instance, by performing binary search. This adds a logarithmic factor that depends on the desired precision. 
The values of $\alphaadj$ obtained using this procedure for selected $k, p$ and $\alpha=0.1$ appear on Table~\ref{tbl:alpha_corrected}.

\begin{table}[t!]
	\caption{Adjusted significance $\alphaadj$ obtained by using \algoCorrect with $\alpha=0.1$ for selected $k, p$. For small values of $k, p$ there is no $\alphaadj$ that yields the required significance.}
	\vspace{-3mm}
	\label{tbl:alpha_corrected}
	\small\begin{tabular}{r|cccc}
		\diaghead{soi text}%
		{$\;$\\p}{k}&
		$40$ & $100$ & $1,000$ & $1,500$ \\\midrule
		0.1 & -- & -- & 0.0140 & 0.0122 \\
		0.2 & -- & -- & 0.0115 & 0.0101 \\
		0.3 & -- & 0.0220 & 0.0103 & 0.0092 \\
		0.4 & -- & 0.0222 & 0.0099 & 0.0088 \\
		0.5 & 0.0313 & 0.0207 & 0.0096 & 0.0084 \\
		0.6 & 0.0321 & 0.0209 & 0.0093 & 0.0085 \\
		0.7 & 0.0293 & 0.0216 & 0.0094 & 0.0084 \\
		\bottomrule
	\end{tabular}
\end{table}

\subsection{Utility}
\label{subsec:individual-fairness}

Our notion of utility reflects the desire to select candidates that are potentially better qualified, and to rank them as high as possible.
In contrast with previous works~\cite{yang2016measuring,celis2017ranking}, we do not assume we know the contribution of having a given candidate at a particular position, but instead base our utility calculation on losses due to non-monotonicity.
The qualifications may have been even proven to be biased against a protected group, as is the case with the COMPAS scores~\cite{angwin_2016_machine} that we use in the experiments of \S\ref{sec:experiments}, but our approach can bound the effect of that bias, because the utility maximization is subject to the ranked group fairness constraint.

\spara{Ranked utility.}
The ranked individual utility associated to a candidate $i$ in a ranking $\tau$, compares it against the least qualified candidate ranked above it.

\begin{definition}[Ranked utility of an element]
	\label{def:rankedIndividualFairness}
	The ranked utility of an element $i \in [n]$ in ranking $\tau$, is:
	\[
	\operatorname{utility}(i,\tau) = \begin{cases}
	\overline{q} - q_i & \textrm{if~} \overline{q} \triangleq \min_{j: r(j,\tau) < r(i,\tau)} q_j < q_i \\
	0 & \textrm{otherwise}\\
	\end{cases}
	\]
\end{definition}
By this definition, the maximum ranked individual utility that can be attained by an element is zero.
%

\spara{Selection utility.}
We operationalize criterion~\ref{cond:selection} of \S\ref{concept:criteria} by means of a \emph{selection utility} objective, which we will use to prefer rankings in which the more qualified candidates are included, and the less qualified, excluded.

\begin{definition}[Selection utility]
	\label{def:selectionFairness}
	The selection utility of a ranking $\tau \in {\mathcal T}_{k,n}$ is $\min_{i \in [n], i \notin \tau} \operatorname{utility}(i,\tau)$.
\end{definition}

Naturally, a ``color-blind'' top-k ranking $c|_k$ maximizes selection utility, i.e., has selection utility zero.

\spara{Ordering utility and in-group monotonicity.}
We operationalize criterion~\ref{cond:ordering} of \S\ref{concept:criteria} by means of an \emph{ordering utility} objective and an \emph{in-group monotonicity constraint}, which we will use to prefer top-$k$ lists in which the more qualified candidates are ranked above the less qualified ones.

\begin{definition}[Ordering utility]
	\label{def:orderingFairness}
	The ordering utility of a ranking $\tau \in {\mathcal T}_{k,n}$ is $\min_{i \in \tau} \operatorname{utility}(i,\tau)$.
\end{definition}

The ordering utility of a ranking is only concerned with the candidate attaining the worst (minimum) ranked individual utility. Instead, the in-group monotonicity constraints refer to all elements, and specifies that both protected and non-protected candidates, independently, must be sorted by decreasing qualifications.

\begin{definition}[In-group monotonicity]
	\label{def:inGroupMonotonicity}
	A ranking $\tau \in {\mathcal T}_{k,n}$ satisfies the in-group monotonicity condition if $\forall i,j$ s.t. $g_i = g_j$, $r(i,\tau) < r(j,\tau) \Rightarrow q_i \ge q_j$.
\end{definition}

Again, the ``color-blind'' top-k ranking $c|_k$ maximizes ordering utility, i.e., has ordering utility zero; it also satisfies the in-group monotonicity constraint.

\spara{Connection to the individual fairness notion.}\label{concept:our-utility-individual-fairness}
Our notion of utility is centered on individuals, for instance by taking the minima instead of averaging.
While other choices are possible, this has the advantage that we can trace loss of utility to specific individuals. These are the people who are ranked below a less qualified candidate, or excluded from the ranking, due to the ranked group fairness constraint.
This is connected to the notion of individual fairness, which requires people to be treated consistently~\cite{Dwork2012}. Under this interpretation, a consistent treatment should require that two people with the same qualifications be treated equally, and any deviation from this is in our framework a utility loss. This allows trade-offs to be made explicit.

\subsection{Formal Problem Statement}
\label{subsec:problem-statement}
The criteria we have described allow for different problem statements, depending on whether we use ranked group fairness as a constraint and maximize ranked utility, or vice versa.

\newtheorem*{problem*}{Problem}
\begin{problem*}[Fair top-k ranking]
	Given a set of candidates $[n]$ and parameters $k$, $p$, and $\alpha$, produce a ranking $\tau \in {\mathcal T}_{k,n}$ that:
	\begin{compactenum}[(i)]
		\item \label{problem:constraint-monotonicity} satisfies the in-group monotonicity constraint;
		\item \label{problem:constraint-rank} satisfies ranked group fairness with parameters $p$ and $\alpha$;
		\item \label{problem:optimal-sel} achieves optimal selection utility subject to (\ref{problem:constraint-monotonicity}) and (\ref{problem:constraint-rank}); and
		\item \label{problem:maximum-ord} maximizes ordering utility subject to (\ref{problem:constraint-monotonicity}), (\ref{problem:constraint-rank}), and (\ref{problem:optimal-sel}).
	\end{compactenum}
\end{problem*}

\spara{Related problems.}\label{concept:related-problems}
Alternative problem definitions are possible with the general criteria described in \S\ref{concept:criteria}.
For instance, instead of maximizing selection and ordering utility, we may seek to keep the utility loss bounded, e.g., producing a ranking that satisfies in-group monotonicity and ranked group fairness, and that produces an $\epsilon$-bounded loss with respect to ordering and/or selection utility.
If the ordering does not matter, we have a {\sc Fair Top-$k$ Selection Problem}, in which we just want to maximize selection utility.
Conversely, if the entire set $[n]$ must be ordered, we have a {\sc Fair Ranking Problem}, in which we just want to maximize ordering utility.
If $k$ is not specified, we have a {\sc Fair Selection Problem}, which resembles a classification problem, and in which the objective might be to  maximize a combination of ranked group fairness, selection utility, and ordering utility.
This multi-objective problem would require a definition of how to combine the different criteria.

\section{Algorithm}\label{sec:algorithms}

We present the \algoFAIR algorithm (\S\ref{subsec:algorithm-description}) and prove it is correct (\S\ref{subsec:algorithm-correctness}).

\subsection{Algorithm Description}\label{subsec:algorithm-description}

\textbf{Algorithm \algoFAIRBF}, presented in Algorithm~\ref{alg:fair}, solves the {\sc Fair Top-$k$ Ranking} problem.
As input, FA*IR takes
the expected size $k$ of the ranking to be returned,
the qualifications $q_i$,
indicator variables $g_i$ indicating if element $i$ is protected,
the target minimum proportion $p$ of protected candidates, and
the adjusted significance level $\alphaadj$.

First, the algorithm uses $q_i$ to create two priority queues with up to $k$ candidates each: $P_0$ for the non-protected candidates and $P_1$ for the protected candidates.
Next (lines \ref{alg:fair:mstart}-\ref{alg:fair:mend}), the algorithm derives a ranked group fairness table similar to Table~\ref{tbl:ranked_group_fairness_table}, i.e., for each position it computes the minimum number of protected candidates, given $p$, $k$ and $\alphaadj$.
Then, \algoFAIR greedily constructs a ranking subject to candidate qualifications, and minimum protected elements required, resembling the method by~\citet{celis2017ranking} for the case of a single protected attribute (the main difference being that we compute the table $m$, while \cite{celis2017ranking} assumes it is given).
If the previously computed table $m$ demands a protected candidate at the current position, the algorithm appends the best candidate from $P_1$ to the ranking (Lines \ref{alg:fair:pstart}-\ref{alg:fair:pend}); otherwise, it appends the best candidate from $P_0 \cup P_1$ (Lines \ref{alg:fair:anystart}-\ref{alg:fair:anyend}).

\algoFAIR has running time $O(n + k \log k)$; which includes building the two $O(k)$ size priority queues from $n$ items and processing them to obtain the ranking, where we assume $k < O(n/\log n)$.
If we already have two ranked lists for both classes of elements, \algoFAIR can avoid the first step and obtain the top-$k$ in $O(k \log k)$ time.
Our method is applicable as long as there is a protected group and there are enough candidates from that group; if there are $k$ from each group, the algorithm is guaranteed to succeed, otherwise the ``head'' of the ranking will satisfy the ranked group fairness constraint, but the ``tail'' of the ranking may not.

\begin{algorithm}[t]
	\caption{Algorithm \algoFAIR finds a ranking that maximizes utility subject to in-group monotonicity and ranked group fairness constraints. Checks for special cases (e.g., insufficient candidates of a class) are not included for clarity.}
	\label{alg:fair}  
	\small
	\AlgInput{$k \in [n]$, the size of the list to return; $\forall~i \in [n]$: $q_i$, the qualifications for candidate $i$, and $g_i$ an indicator that is 1 iff candidate $i$ is protected; $p \in ]0,1[$, the minimum proportion of protected elements; $\alphaadj \in ]0,1[$, the adjusted significance for each fair representation test.}
	\AlgOutput{$\tau$ satisfying the group fairness condition with parameters $p, \sigma$, and maximizing utility.}
	$P_0, P_1 \leftarrow$ empty priority queues with bounded capacity $k$\\
	\For{$i \leftarrow 1$ \KwTo $n$}{
		insert $i$ with value $q_i$ in priority queue $P_{g_i}$ \\
	}
	\For{$i \leftarrow 1$ \KwTo $k$}{ \label{alg:fair:mstart}
		$m[i] \leftarrow F^{-1}(\alphaadj; i, p)$ \\
	}\label{alg:fair:mend}
	$(t_p, t_n) \leftarrow (0, 0)$ \\
	\While{$t_p + t_n < k$}{
		\eIf{$t_p < m[t_p + t_n + 1]$}{
			\AlgComment{add a protected candidate}
			$t_p \leftarrow t_p + 1$ \\ \label{alg:fair:pstart}
			$\tau[t_p + t_n] \leftarrow \operatorname{dequeue}(P_1)$ \\  \label{alg:fair:pend}
		}{
		\AlgComment{add the best candidate available}
		\eIf{$q(\operatorname{peek}(P_1)) \ge q(\operatorname{peek}(P_0))$} { \label{alg:fair:anystart}
			$t_p \leftarrow t_p + 1$ \\
			$\tau[t_p + t_n] \leftarrow \operatorname{dequeue}(P_1)$ \\

		}{
		$t_n \leftarrow t_n + 1$ \\
		$\tau[t_p + t_n] \leftarrow \operatorname{dequeue}(P_0)$ \\
	} \label{alg:fair:anyend}
}

}
\Return{$\tau$}
\end{algorithm}
\vspace{-3mm}

\subsection{Algorithm Correctness}\label{subsec:algorithm-correctness}

By construction, a ranking $\tau$ generated by \algoFAIR satisfies in-group monotonicity, because protected and non-protected candidates are selected by decreasing qualifications.
It also satisfies the ranked group fairness constraint, because for every prefix of size $i$ the list, the number of protected candidates is at least $m[i]$.
What we must prove is that $\tau$ achieves optimal selection utility, and that it maximizes ordering utility.
This is done in the following lemmas.

\begin{lemma}\label{lemma:across}
	If a ranking satisfies the in-group monotonicity constraint, then the utility loss (ordering or selection utility different from zero) can only happen across protected/non-protected groups.
\end{lemma}

\begin{proof}
	This comes directly from Definition~\ref{def:inGroupMonotonicity} given that for two elements $i,j$, the only case in which $r(i,\tau) < r(j,\tau) \wedge q_i < q_j$ is when $g_i \ne g_j$.
\end{proof}

\begin{lemma}
	The optimal selection utility among rankings satisfying in-group monotonicity (\ref{problem:constraint-monotonicity}) and ranked group fairness (\ref{problem:constraint-rank}), is either zero, or is due to a non-protected candidate ranked below a less qualified protected candidate.
\end{lemma}

\begin{proof}
	Let $i,j$ be the two elements that attain the optimal selection utility, with $i \in \tau, j \in [n] \backslash \tau$.
	We will prove by contradiction: let us assume $i$ is a non-protected element ($g_i=0$) and $j$ is a protected element ($g_j=1$).
	By in-group monotonicity, we know $i$ is the last non-protected element in $\tau$. Let us swap $i$ and $j$, moving $i$ outside $\tau$ and $j$ inside the ranking, and then moving down $j$ if necessary to place it in the correct ordering among the protected elements below its position (given that $i$ is the last non-protected element in $\tau$).
	The new ranking continues to satisfy in-group monotonicity as well as ranked group fairness (as it has not decreased the number of protected elements at any position in the ranking), and has a larger selection utility.
	This is a contradiction because the selection utility was optimal. Hence, $i$ is a protected element and $j$ a non-protected element.
\end{proof}

\begin{lemma}\label{lemma:number-protected-implies-selfairness} 
	Given two rankings $\rho, \tau$ satisfying in-group monotonicity (\ref{problem:constraint-monotonicity}), if they have the same number of protected elements $\rho_p = \tau_p$, then both rankings contain the same $k$ elements (possibly in different order), and hence both rankings have the same selection utility.
\end{lemma}

\begin{proof}
	Both rankings contain a prefix of size $\tau_p$ of the list of protected candidates ordered by decreasing qualifications, and a prefix of size $k - \tau_p$ of the list of non-protected candidates ordered by decreasing qualifications.
	Hence, $\forall i \in [n], i \in \tau \Leftrightarrow i \in \rho$, so the elements not included in the rankings are also the same elements, and the selection utility of both rankings is the same.
\end{proof}

The previous lemma means selection utility is determined by the number of protected candidates in a ranking.

\begin{lemma}\label{lemma:fair-optimal-selection}
	Algorithm \algoFAIR achieves optimal selection utility among rankings satisfying in-group monotonicity (\ref{problem:constraint-monotonicity}) and ranked group fairness (\ref{problem:constraint-rank}).
\end{lemma}

\begin{proof}
	Let $\tau$ be the ranking produced by \algoFAIR, and $\tau^*$ be the ranking achieving the optimal selection utility. We will prove that $\tau_p = \tau^*_p$ by contradiction.
	Suppose $\tau_p < \tau^*_p$. Then, we could take the least qualified protected element in $\tau^*_p$ and swap it with the most qualified non-protected element in $[n] \backslash \tau^*_p$, re-ordering as needed. This would increase selection utility and still satisfy the constraints, which is a contradiction with the fact that $\tau^*_p$ achieved the optimal selection utility.
	Suppose $\tau_p > \tau^*_p$. Then, at the position at which the least qualified protected element in $\tau$ is found, we could have placed a non-protected element with higher qualifications, as $\tau^*$ satisfies ranked group fairness and has less protected elements. This is a contradiction with the way in which \algoFAIR operates, as it only places a protected element with lower qualifications when needed to satisfy ranked group fairness.
	Hence, $\tau_p = \tau^*_p$ and by Lemma~\ref{lemma:number-protected-implies-selfairness} it achieves the same selection utility.
\end{proof}

\begin{lemma}
	Algorithm \algoFAIR maximizes ordering utility among rankings satisfying in-group monotonicity (\ref{problem:constraint-monotonicity}), ranked group fairness (\ref{problem:constraint-rank}), and achieving optimal selection utility (\ref{problem:optimal-sel}).
\end{lemma}

\begin{proof}
	By lemmas~\ref{lemma:number-protected-implies-selfairness} and \ref{lemma:fair-optimal-selection} we know that satisfying the constraints and achieving optimal selection utility implies having a specific number of protected elements $\tau^*_p$.
	Hence, we need to show that among rankings having this number of protected elements, \algoFAIR achieves the maximum ordering utility.
	By Lemma~\ref{lemma:across} we know that loss of ordering utility is due only to non-protected elements placed below less qualified protected elements.
	However, we know that in \algoFAIR this only happens when necessary to satisfy ranked group fairness, and having less protected elements at any given position than the ranking produced by \algoFAIR would violate the ranked group fairness constraint.
\end{proof}

\section{Experiments}\label{sec:experiments}

In the first part of our experiments we create synthetic datasets to demonstrate the correctness of the adjustment done by Algorithm \algoCorrect (\S\ref{subsubsec:JuliaExperimentalVerification}).
In the second part, we consider several public datasets, as well as new datasets that we make public, for evaluating algorithm \algoFAIR (datasets in \S\ref{sec:experiments-datasets}, metrics and comparison with baselines in \S\ref{sec:experiments-baselines}, and results in \S\ref{sec:experiments-results}).

\subsection{Verification of Multiple Tests Adjustment}
\label{subsubsec:JuliaExperimentalVerification}

We empirically verified the adjustment formula and the \algoCorrect method using randomly generated data.
We repeatedly generated multiple rankings of different lengths $k$ using the algorithm by \citet{yang2016measuring} and evaluated these rankings with our ranked group fairness test, determining the probability that this ranking, which we consider fair, was declared unfair.
Example results are shown on Figure~\ref{fig:julia-experimental-verification} for some combinations of $k$ and $\alphaadj$.
As expected, the experiment results closely resemble the output of \algoCorrect.

\begin{figure}[t]
	\includegraphics[width=.55\columnwidth]{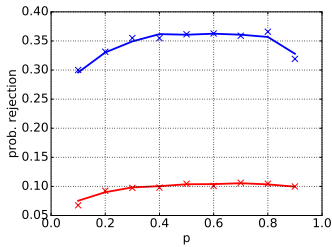}
	\vspace{-2mm}
	\caption{Probability of considering a fair ranking generated by~\cite{yang2016measuring} as unfair for $k=1,000; \alphaadj=0.01$ (bottom curve) and for $k=1,500; \alphaadj=0.05$ (top curve). Model represented by lines, experimental results (avg. of 10,000 runs) by crosses.}
	\vspace{-5mm}
	\label{fig:julia-experimental-verification}
\end{figure}

\begin{table}[t]
	\caption{Datasets and experimental settings.}
	\vspace{-3mm}
	\label{tbl:datasets}
	\resizebox{1.01\columnwidth}{!}{%
		\centering\begin{tabular}{clcccccc}\toprule
			&        &                         &                         & Quality   & Protected & Protected \\
			& Dataset & \multicolumn{1}{c}{$n$} & \multicolumn{1}{c}{$k$} & criterion & group     & \% \\ \midrule
			D1 & COMPAS \cite{angwin_2016_machine}& 18K  & 1K & $\neg$recidivism & Afr.-Am. & 51.2\% \\
			D2 & "  & "  & " & " & male & 80.7\%\\
			D3 & "  & "  & " & " & female & 19.3\%\\
			D4 & Ger. credit \cite{lichman_2013_uci} & 1K   & 100   & credit rating & female & 69.0\% \\
			D5 & " & " & " & " & $<$ 25 yr. & 14.9\% \\
			D6 & "  & " & " & " & $<$ 35 yr. & 54.8\% \\
			D7 & SAT \cite{sat_2014}   & 1.6 M & 1.5K  & test score  & female & 53.1\%  \\
			D8 & XING [ours]           & 40 & 40 & ad-hoc score  & f/m/f & 27/43/27\%  \\
			\bottomrule
		\end{tabular}
	}
	\vspace{-3mm}
\end{table}

\subsection{Datasets}\label{sec:experiments-datasets}

Table~\ref{tbl:datasets} summarizes the datasets used in our experiments.
Each dataset contains a set of people with demographic attributes, plus a quality attribute.
For each dataset, we consider a value of $k$ that is a small round number  ({\em e.g.}, 100, 1,000, or 1,500), or $k=n$ for a small dataset.
For the purposes of these experiments, we considered several scenarios of protected groups.
We remark that the choice of protected group is not arbitrary: it is determined completely by law or voluntary commitments; for the purpose of experimentation we test different scenarios, but in a real application there is no ambiguity about which is the protected group and what is the minimum proportion.
An experiment consists of generating a ranking using \algoFAIR and then comparing it with baseline rankings according to the metrics introduced in the next section.

We used the two publicly-available datasets used in \cite{yang2016measuring} (COMPAS~\cite{angwin_2016_machine} and German Credit~\cite{lichman_2013_uci}), plus another publicly available dataset (SAT~\cite{sat_2014}), plus a new dataset created and released with this paper (XING), as we describe next.

\spara{COMPAS} (Correctional Offender Management Profiling for Alternative Sanctions) is an assessment tool for predicting recidivism based on a questionnaire of 137 questions. It is used in several jurisdictions in the US, and has been accused of racial discrimination by producing a higher likelihood to recidivate for African Americans~\cite{angwin_2016_machine}.
In our experiment, we test a scenario in which we want to create a fair ranking of the top-$k$ people who are least likely to recidivate, who could be, for instance, considered for a pardon or reduced sentence.
We observe that African Americans as well as males are given a larger recidivism score than other groups; for the purposes of this experiment we select these two categories as the protected groups.

\spara{German Credit} is the Statlog German Credit Data collected by Hans Hofmann~\cite{lichman_2013_uci}.
It is based on credit ratings generated by Schufa, a German private credit agency based on a set of variables for each applicant, including age, gender, marital status, among others. Schufa Score is an essential determinant for every resident in Germany when it comes to evaluating credit rating before getting a phone contract, a long-term apartment rental or almost any loan.
We use the credit-worthiness as qualification, as~\cite{yang2016measuring}, and note that women and younger applicants are given lower scores; for the purposes of these experiments, we use those groups as protected.

\spara{SAT} corresponds to scores in the US Scholastic Assessment Test, a standardized test used for college admissions in the US.
We generate this data using the actual distribution of SAT results from 2014, which is publicly available for 1.6 million applicants in fine-grained buckets of 10 points (out of a total of 2,400 points)~\cite{sat_2014}.
The qualification attribute is set to be the achieved SAT score, and the protected group is women (female students), who scored about 25 points lower on average than men in this test.

\spara{XING} (\url{https://www.xing.com/}) is a career-oriented website from which we automatically collected the top-40 profiles returned for 54 queries, using three for which there is a clear difference between top-10 and top-40.
We used a non-personalized (not logged in) search interface and confirmed that it yields the same results from different locations.\label{concept:XING}
For each profile, we collected gender, list of positions held, list of education details, and the number of times each profile has been viewed in the platform, which is a measure of popularity of the profile.
With this information, we constructed an ad-hoc score:
the months of work experience 
plus the months of education, %
multiplied by the number of views of the profile.
This score tends to be somewhat higher for profiles in the first positions of the search results, but in general does not approximate the proprietary ordering in which profiles are shown. 
We include this score and its components in our anonymized data release.
We use the appropriate gender for each query as the protected group.

\subsection{Baselines and Metrics}\label{sec:experiments-baselines}

For each dataset, we generate various top-$k$ rankings with varying targets of minimum proportion of protected candidates $p$ using \algoFAIR, plus two baseline rankings:

\spara{Baseline 1: Color-blind ranking.} The ranking $c|_k$ that only considers the qualifications of the candidates, without considering group fairness, as described in Section~\ref{concept:color-blind-ranking}.

\spara{Baseline 2: \citet{Feldman2015}.} This ranking method aligns the probability distribution of the protected candidates with the non-protected ones. Specifically, for a candidate $i$ in the protected group, we replace its score $q_i \leftarrow q_j$ by choosing a candidate $j$ in the non-protected group having $F_n(j) = F_p(i)$, with $F_p(\cdot)$ (respectively, $F_n(\cdot))$ being the quantile of a candidate among the protected (respectively, non-protected) candidates.

\spara{Utility.} We report the loss in ranked utility after score normalization, in which all $q_i$ are normalized to be within $[0, 1]$.
We also report the maximum rank drop, {\em i.e.}, the number of positions lost by the candidate that realizes the maximum ordering utility loss.

\spara{NDCG.}
We report a normalized weighted summation of the quality of the elements in the ranking, $\sum_{i=1}^{k} w_i q_{(\tau_i)}$, in which the weights are chosen to have a logarithmic discount in the position:  $w_i = \frac{1}{\log_2 (i+1)}$. This is a standard measure to evaluate search rankings~\cite{jarvelin2002cumulated}.
This is normalized so that the maximum value is $1.0$.

\subsection{Results}\label{sec:experiments-results}

\begin{table}[t]
	\caption{Experimental results, highlighting in boldface the best non-color-blind result. Both FA*IR and the baseline from \citeauthor{Feldman2015} achieve the same target proportion of protected elements in the output and the same selection unfairness, but in general FA*IR achieves it with less ordering unfairness, and with less maximum rank drop (the number of positions that the most unfairly ordered element drops).}
	\vspace{-3mm}
	\label{tbl:results}
	\resizebox{1.01\columnwidth}{!}{%
		\centering\begin{tabular}{llcccccc}\toprule
			&        & \% Prot. &       & Ordering     & Rank & Selection \\
			& Method & output   & NDCG  & utility loss & drop & utility loss \\ \midrule
			D1 (51.2\%) & Color-blind & 25\% & 1.0000 & 0.0000 & 0 & 0.0000 \\
			COMPAS, & FA*IR p=0.5 & 46\% & \textbf{0.9858} & \textbf{0.2026} & \textbf{319} & \textbf{0.1087} \\
			race$=$Afr.-Am. & \citeauthor{Feldman2015} & 51\% & 0.9779 & 0.2281 & 393 & 0.1301 \\ \midrule

			D2 (80.7\%) & Color-blind & 73\% & 1.0000 & 0.0000 & 0 & 0.0000 \\
			COMPAS, & FA*IR p=0.8 & 77\% & \textbf{1.0000} & \textbf{0.1194} & \textbf{161} & \textbf{0.0320} \\
			gender$=$male & \citeauthor{Feldman2015} & 81\% & 0.9973 & 0.2090 & 294 & 0.0533 \\ \midrule

			D3 (19.3\%) & Color-blind & 28\% & 1.0000 & 0.0000 & 0 & 0.0000 \\
			COMPAS, & FA*IR p=0.2 & 28\% & \textbf{0.9999} & \textbf{0.2239} & \textbf{1} & \textbf{0.0000} \\
			gender$=$female & \citeauthor{Feldman2015} & 19\% & 0.9972 & 0.3028 & 278 & 0.0533 \\ \midrule

			D4 (69.0\%) & Color-blind & 74\% & 1.0000 & 0.0000 & 0 & 0.0000 \\
			Ger. cred, & FA*IR p=0.7 & 74\% & \textbf{1.0000} & \textbf{0.0000} & \textbf{0} & \textbf{0.0000} \\
			gender$=$female & \citeauthor{Feldman2015} & 69\% & 0.9988 & 0.1197 & 8 & 0.0224 \\ \midrule

			D5 (14.9\%) & Color-blind & 9\% & 1.0000 & 0.0000 & 0 & 0.0000 \\
			Ger. cred,& FA*IR p=0.2 & 15\% & \textbf{0.9983} & \textbf{0.0436} & \textbf{7} & \textbf{0.0462} \\
			age $<$ 25 & \citeauthor{Feldman2015} & 15\% & 0.9952 & 0.1656 & 8 & \textbf{0.0462} \\ \midrule

			D6 (54.8\%) & Color-blind & 24\% & 1.0000 & 0.0000 & 0 & 0.0000 \\
			Ger. cred, & FA*IR p=0.6 & 50\% & \textbf{0.9913} & \textbf{0.1137} & \textbf{30} & \textbf{0.0593} \\
			age $<$ 35 & \citeauthor{Feldman2015} & 55\% & 0.9853 & 0.2123 & 36 & 0.0633 \\ \midrule

			D7 (53.1\%) & Color-blind    & 49\% & 1.0000 & 0.0000 & 0 & 0.0000 \\
			SAT, 	    & FA*IR p=0.6    & 57\% & \textbf{0.9996} & \textbf{0.0167} & 365 & 0.0083 \\
			gender$=$female & \citeauthor{Feldman2015} & 56\% & \textbf{0.9996} & \textbf{0.0167} & \textbf{241} & \textbf{0.0042} \\ \midrule

			D8a (27.5\%)  & Color-blind    & 28\% & 1.0000 		& 0.0000 	& 0  & 0.0000 \\
			Economist,     & FA*IR p=0.3    & 28\% & \textbf{1.0000} & \textbf{0.0000} & \textbf{0}  & \textbf{0.0000} \\
			gender$=$female & \citeauthor{Feldman2015} & 28\% & 0.9935 		& 0.6109 	& 5 & \textbf{0.0000} \\ \midrule
			D8b (42.5\%)  & Color-blind    & 43\% & 1.0000 		& 0.0000 	& 0           & 0.0000 \\
			Mkt. Analyst,  & FA*IR p=0.4    & 43\% & \textbf{1.0000} & \textbf{0.0000} & \textbf{0}  & \textbf{0.0000} \\
			gender$=$male & \citeauthor{Feldman2015} & 43\% & 0.9422 		& 1.0000 	& 5 & \textbf{0.0000} \\ \midrule
			D8c (29.7\%)  & Color-blind    & 30\% & 1.0000 		& 0.0000 	& 0           & 0.0000 \\
			Copywriter,    & FA*IR p=0.3    & 30\% & \textbf{1.0000} & \textbf{0.0000} & \textbf{0}  & \textbf{0.0000} \\
			gender$=$female & \citeauthor{Feldman2015} & 30\% & 0.9782 & 0.4468 & 10 & \textbf{0.0000} \\
			\bottomrule
		\end{tabular}
	}
	\vspace{-3mm}
\end{table}

Table~\ref{tbl:results} summarizes the results. We report on the result using $p$ as a multiple of $0.1$ close to the proportion of protected elements in each dataset.
First, we observe that in general changes in utility with respect to the color-blind ranking are minor, as the utility is dominated by the top positions, which do not change dramatically.
Second, \algoFAIR achieves higher or equal selection utility than the baseline~\cite{Feldman2015} in all but one of the experimental conditions (D7).
Third, \algoFAIR achieves higher or equal ordering utility in all conditions. This is also reflected in the rank loss of the most unfairly
treated candidate included in the ranking ({\em i.e.}, the candidate that achieves the maximum ordering utility loss). 

Interestingly, \algoFAIR allows to create rankings for multiple values of $p$, something that cannot be done directly with the baselines (\citet{Feldman2015} allows what they call a ``partial repair,'' but through an indirect parameter determining a mixture of the original and a transformed distribution).
Figure~\ref{fig:results-moving-p} shows results when varying $p$ in dataset D4 (German credit, the protected group is people under 25 years old).
This means that \algoFAIR allows a wide range of positive actions, for instance, offering favorable credit conditions to people with good credit rating, with a preference towards younger customers.
In this case, the figure shows that we can double the proportion of young people in the top-$k$ ranking (from the original 15\% up to 30\%) without introducing a large ordering utility loss and maintaining NDCG almost unchanged.

\begin{figure}[t!]
	\centering
	\subfigure[Ordering utility]{\includegraphics[width=.49\columnwidth]{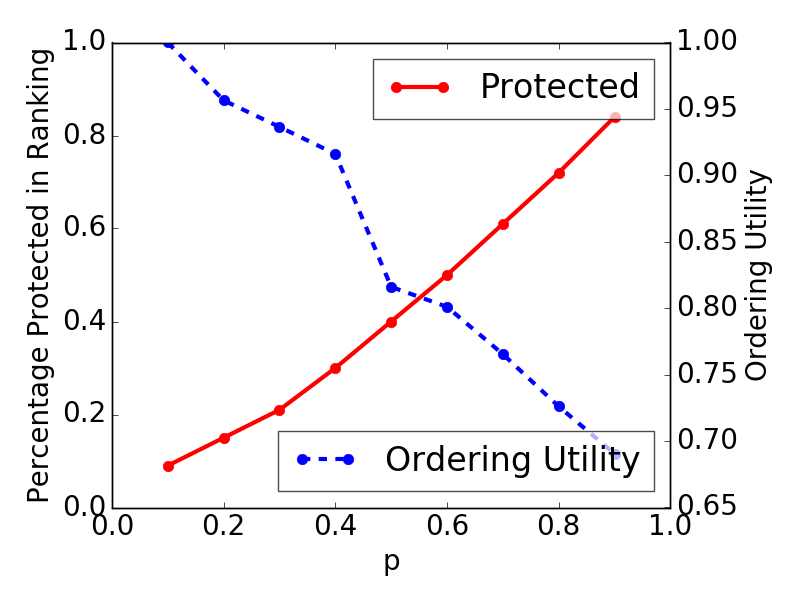}}
	\subfigure[NDCG]{\includegraphics[width=.49\columnwidth]{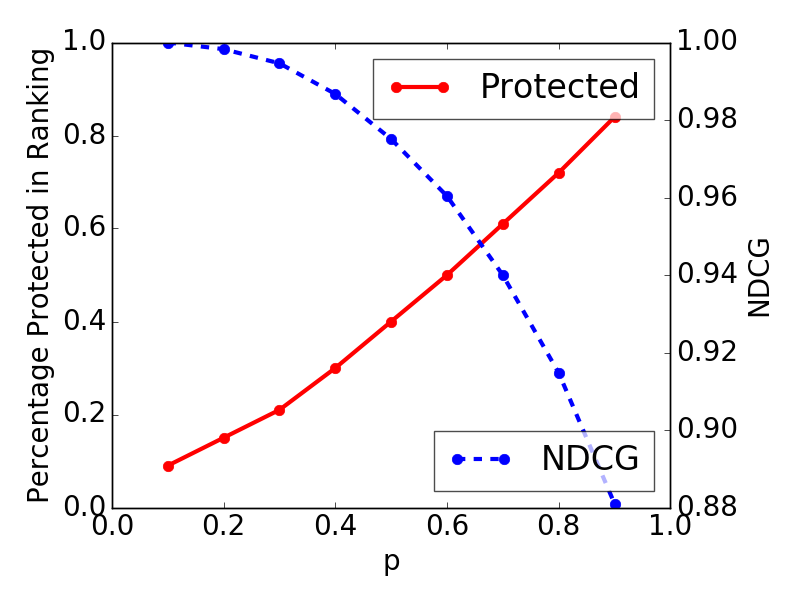}}
	\vspace{-2mm}
	\caption{Depiction of possible trade-offs using \algoFAIR. Increase in the percentage of protected candidates in D5 (German credit, protected group age $<$ 25) for increasing values of $p$, compared to 
		ordering utility and NDCG.}
	\vspace{-\baselineskip}
	\label{fig:results-moving-p}
\end{figure}

\section{Conclusions}\label{sec:conclusions}

The method we have presented can generate a ranking with a guaranteed ranked group fairness, and as we have observed, does not introduce a large utility loss.
Compared to the baseline of~\citet{Feldman2015}, in general we introduce the same or less utility loss. We also do not assume that the distributions of qualifications in the protected and non-protected groups have a similar shape.
More importantly, we can directly control through a parameter $p$ the trade-off between fairness and utility.

\spara{Future work.}
For simplicity, we have considered a situation where people belong to a protected or a non-protected group, and leave the case of multiple protected groups or combinations of protected attributes for future work; we plan to adapt methods based on linear programming to achieve this~\cite{celis2017ranking}.
We are also experimenting with the related problems we considered in \S\ref{concept:related-problems}, including directly bounding the maximum utility loss (ordering or selection), while maximizing ranked group fairness, or weighing the three criteria.

One of the main challenges is to create an in-processing ranking method instead of a post-processing one. However, we must also be cautious as results by \citet{kleinberg2016inherent} stating that one cannot have a predictor of risk that is well calibrated and satisfies statistical parity requirements, may imply that having a fair ranking by construction is not possible. We should also consider explainable discrimination~\cite{vzliobaite2011handling}, or even try to show a causal relation between protected attributes and qualification scores.

Experimentally, there are several directions. For instance, we have used real datasets that exhibit some differences among protected and non-protected groups; experiments with synthetic datasets would allow to test with more extreme differences that are more rarely found in real data.
Further experimental work may be done to measure robustness to noise in qualifications, and later to evaluate the impact of this in a real application.

\spara{Reproducibility.}
All the code and data used on this paper is available online at \url{https://github.com/MilkaLichtblau/FA-IR_Ranking}.

\section{Acknowledgements}

This research was supported by the German Research Foundation, Eurecat and the Catalonia Trade
and Investment Agency (ACCI{\'O}). M.Z. and M.M. were supported by the GRF. C.C. and S.H. worked
on this paper while at Eurecat. C.C., S.H., and F.B. were supported by ACCI{\'O}. We would also
like to show our gratitude to the three anonymous reviewers for their comments and insights.


	\balance


\begin{thebibliography}{00}


\ifx \showCODEN    \undefined \def \showCODEN     #1{\unskip}     \fi
\ifx \showDOI      \undefined \def \showDOI       #1{{\tt DOI:}\penalty0{#1}\ }
  \fi
\ifx \showISBNx    \undefined \def \showISBNx     #1{\unskip}     \fi
\ifx \showISBNxiii \undefined \def \showISBNxiii  #1{\unskip}     \fi
\ifx \showISSN     \undefined \def \showISSN      #1{\unskip}     \fi
\ifx \showLCCN     \undefined \def \showLCCN      #1{\unskip}     \fi
\ifx \shownote     \undefined \def \shownote      #1{#1}          \fi
\ifx \showarticletitle \undefined \def \showarticletitle #1{#1}   \fi
\ifx \showURL      \undefined \def \showURL       #1{#1}          \fi
\providecommand\bibfield[2]{#2}
\providecommand\bibinfo[2]{#2}
\providecommand\natexlab[1]{#1}
\providecommand\showeprint[2][]{arXiv:#2}

\bibitem[\protect\citeauthoryear{Angwin, Larson, Mattu, and Kirchner}{Angwin
  et~al\mbox{.}}{2016}]%
        {angwin_2016_machine}
\bibfield{author}{\bibinfo{person}{Julia Angwin}, \bibinfo{person}{Jeff
  Larson}, \bibinfo{person}{Surya Mattu}, {and} \bibinfo{person}{Lauren
  Kirchner}.} \bibinfo{year}{2016}\natexlab{}.
\newblock \showarticletitle{Machine Bias}.
\newblock \bibinfo{journal}{{\em ProPublica\/}} (\bibinfo{date}{May}
  \bibinfo{year}{2016}).
\newblock


\bibitem[\protect\citeauthoryear{Barocas and Selbst}{Barocas and
  Selbst}{2014}]%
        {Barocas2014}
\bibfield{author}{\bibinfo{person}{Solon Barocas} {and}
  \bibinfo{person}{Andrew~D Selbst}.} \bibinfo{year}{2014}\natexlab{}.
\newblock \showarticletitle{Big data's disparate impact}.
\newblock \bibinfo{journal}{{\em SSRN 2477899\/}} (\bibinfo{year}{2014}).
\newblock


\bibitem[\protect\citeauthoryear{Bonchi, Hajian, Mishra, and Ramazzotti}{Bonchi
  et~al\mbox{.}}{2017}]%
        {Bonchi2015}
\bibfield{author}{\bibinfo{person}{Francesco Bonchi}, \bibinfo{person}{Sara
  Hajian}, \bibinfo{person}{Bud Mishra}, {and} \bibinfo{person}{Daniele
  Ramazzotti}.} \bibinfo{year}{2017}\natexlab{}.
\newblock \showarticletitle{Exposing the probabilistic causal structure of
  discrimination}.
\newblock \bibinfo{journal}{{\em International Journal of Data Science and
  Analytics\/}} \bibinfo{volume}{3}, \bibinfo{number}{1}
  (\bibinfo{year}{2017}), \bibinfo{pages}{1--21}.
\newblock


\bibitem[\protect\citeauthoryear{Calders and Verwer}{Calders and
  Verwer}{2010}]%
        {Calders2010}
\bibfield{author}{\bibinfo{person}{Toon Calders} {and} \bibinfo{person}{Sicco
  Verwer}.} \bibinfo{year}{2010}\natexlab{}.
\newblock \showarticletitle{Three naive Bayes approaches for
  discrimination-free classification}.
\newblock \bibinfo{journal}{{\em Data Mining and Knowledge Discovery\/}}
  \bibinfo{volume}{21}, \bibinfo{number}{2} (\bibinfo{year}{2010}),
  \bibinfo{pages}{277--292}.
\newblock


\bibitem[\protect\citeauthoryear{Carbonell and Goldstein}{Carbonell and
  Goldstein}{1998}]%
        {carbonell1998use}
\bibfield{author}{\bibinfo{person}{Jaime Carbonell} {and} \bibinfo{person}{Jade
  Goldstein}.} \bibinfo{year}{1998}\natexlab{}.
\newblock \showarticletitle{The use of MMR, diversity-based reranking for
  reordering documents and producing summaries}. In \bibinfo{booktitle}{{\em
  Proc. of {SIGIR}}}. \bibinfo{publisher}{ACM Press},
  \bibinfo{pages}{335--336}.
\newblock


\bibitem[\protect\citeauthoryear{Celis, Deshpande, Kathuria, and Vishnoi}{Celis
  et~al\mbox{.}}{2016}]%
        {celis2016fair}
\bibfield{author}{\bibinfo{person}{L~Elisa Celis}, \bibinfo{person}{Amit
  Deshpande}, \bibinfo{person}{Tarun Kathuria}, {and}
  \bibinfo{person}{Nisheeth~K Vishnoi}.} \bibinfo{year}{2016}\natexlab{}.
\newblock \showarticletitle{How to be Fair and Diverse?}
\newblock \bibinfo{journal}{{\em arXiv:1610.07183\/}} (\bibinfo{year}{2016}).
\newblock


\bibitem[\protect\citeauthoryear{Celis, Straszak, and Vishnoi}{Celis
  et~al\mbox{.}}{2017}]%
        {celis2017ranking}
\bibfield{author}{\bibinfo{person}{L~Elisa Celis}, \bibinfo{person}{Damian
  Straszak}, {and} \bibinfo{person}{Nisheeth~K Vishnoi}.}
  \bibinfo{year}{2017}\natexlab{}.
\newblock \showarticletitle{Ranking with Fairness Constraints}.
\newblock \bibinfo{journal}{{\em arXiv:1704.06840\/}} (\bibinfo{year}{2017}).
\newblock


\bibitem[\protect\citeauthoryear{Channamsetty and Ekstrand}{Channamsetty and
  Ekstrand}{2017}]%
        {channamsetty2017recommender}
\bibfield{author}{\bibinfo{person}{Sushma Channamsetty} {and}
  \bibinfo{person}{Michael~D Ekstrand}.} \bibinfo{year}{2017}\natexlab{}.
\newblock \showarticletitle{Recommender Response to Diversity and Popularity
  Bias in User Profiles (short paper)}. In \bibinfo{booktitle}{{\em 30th
  International Florida Artificial Intelligence Research Society Conference}}.
\newblock


\bibitem[\protect\citeauthoryear{Corbett-Davies, Pierson, Feller, Goel, and
  Huq}{Corbett-Davies et~al\mbox{.}}{2017}]%
        {corbett2017algorithmic}
\bibfield{author}{\bibinfo{person}{Sam Corbett-Davies}, \bibinfo{person}{Emma
  Pierson}, \bibinfo{person}{Avi Feller}, \bibinfo{person}{Sharad Goel}, {and}
  \bibinfo{person}{Aziz Huq}.} \bibinfo{year}{2017}\natexlab{}.
\newblock \showarticletitle{Algorithmic decision making and the cost of
  fairness}.
\newblock \bibinfo{journal}{{\em arXiv:1701.08230\/}} (\bibinfo{year}{2017}).
\newblock


\bibitem[\protect\citeauthoryear{Dwork, Hardt, Pitassi, Reingold, and
  Zemel}{Dwork et~al\mbox{.}}{2012}]%
        {Dwork2012}
\bibfield{author}{\bibinfo{person}{Cynthia Dwork}, \bibinfo{person}{Moritz
  Hardt}, \bibinfo{person}{Toniann Pitassi}, \bibinfo{person}{Omer Reingold},
  {and} \bibinfo{person}{Richard Zemel}.} \bibinfo{year}{2012}\natexlab{}.
\newblock \showarticletitle{Fairness through awareness}. In
  \bibinfo{booktitle}{{\em Proc. of {ITCS}}}. \bibinfo{publisher}{ACM Press},
  \bibinfo{pages}{214--226}.
\newblock


\bibitem[\protect\citeauthoryear{Ellis and Watson}{Ellis and Watson}{2012}]%
        {ellis2012eu}
\bibfield{author}{\bibinfo{person}{Evelyn Ellis} {and}
  \bibinfo{person}{Philippa Watson}.} \bibinfo{year}{2012}\natexlab{}.
\newblock \bibinfo{booktitle}{{\em EU anti-discrimination law}}.
\newblock \bibinfo{publisher}{Oxford University Press}.
\newblock


\bibitem[\protect\citeauthoryear{Feldman, Friedler, Moeller, Scheidegger, and
  Venkatasubramanian}{Feldman et~al\mbox{.}}{2015}]%
        {Feldman2015}
\bibfield{author}{\bibinfo{person}{Michael Feldman}, \bibinfo{person}{Sorelle~A
  Friedler}, \bibinfo{person}{John Moeller}, \bibinfo{person}{Carlos
  Scheidegger}, {and} \bibinfo{person}{Suresh Venkatasubramanian}.}
  \bibinfo{year}{2015}\natexlab{}.
\newblock \showarticletitle{Certifying and removing disparate impact}. In
  \bibinfo{booktitle}{{\em Proc. of {KDD}}}. \bibinfo{publisher}{ACM Press},
  \bibinfo{pages}{259--268}.
\newblock


\bibitem[\protect\citeauthoryear{Friedler, Scheidegger, and
  Venkatasubramanian}{Friedler et~al\mbox{.}}{2016}]%
        {friedler2016possibility}
\bibfield{author}{\bibinfo{person}{Sorelle~A Friedler}, \bibinfo{person}{Carlos
  Scheidegger}, {and} \bibinfo{person}{Suresh Venkatasubramanian}.}
  \bibinfo{year}{2016}\natexlab{}.
\newblock \showarticletitle{On the (im) possibility of fairness}.
\newblock \bibinfo{journal}{{\em arXiv:1609.07236\/}} (\bibinfo{year}{2016}).
\newblock


\bibitem[\protect\citeauthoryear{Friedman and Nissenbaum}{Friedman and
  Nissenbaum}{1996}]%
        {friedman1996bias}
\bibfield{author}{\bibinfo{person}{Batya Friedman} {and} \bibinfo{person}{Helen
  Nissenbaum}.} \bibinfo{year}{1996}\natexlab{}.
\newblock \showarticletitle{Bias in computer systems}.
\newblock \bibinfo{journal}{{\em ACM Transactions on Information Systems\/}}
  \bibinfo{volume}{14}, \bibinfo{number}{3} (\bibinfo{year}{1996}),
  \bibinfo{pages}{330--347}.
\newblock


\bibitem[\protect\citeauthoryear{Hajian, Bonchi, and Castillo}{Hajian
  et~al\mbox{.}}{2016}]%
        {tuto2016}
\bibfield{author}{\bibinfo{person}{Sara Hajian}, \bibinfo{person}{Francesco
  Bonchi}, {and} \bibinfo{person}{Carlos Castillo}.}
  \bibinfo{year}{2016}\natexlab{}.
\newblock \showarticletitle{Algorithmic Bias: From Discrimination Discovery to
  Fairness-aware Data Mining}. In \bibinfo{booktitle}{{\em {KDD} Tutorials}}.
\newblock


\bibitem[\protect\citeauthoryear{Hajian and Domingo-Ferrer}{Hajian and
  Domingo-Ferrer}{2013}]%
        {HajianFerrer12}
\bibfield{author}{\bibinfo{person}{Sara Hajian} {and} \bibinfo{person}{Josep
  Domingo-Ferrer}.} \bibinfo{year}{2013}\natexlab{}.
\newblock \showarticletitle{A methodology for direct and indirect
  discrimination prevention in data mining}.
\newblock \bibinfo{journal}{{\em IEEE TKDE\/}} \bibinfo{volume}{25},
  \bibinfo{number}{7} (\bibinfo{year}{2013}).
\newblock


\bibitem[\protect\citeauthoryear{Hajian, Domingo-Ferrer, and Farr{\`a}s}{Hajian
  et~al\mbox{.}}{2014}]%
        {hajian2014}
\bibfield{author}{\bibinfo{person}{Sara Hajian}, \bibinfo{person}{Josep
  Domingo-Ferrer}, {and} \bibinfo{person}{Oriol Farr{\`a}s}.}
  \bibinfo{year}{2014}\natexlab{}.
\newblock \showarticletitle{Generalization-based privacy preservation and
  discrimination prevention in data publishing and mining}.
\newblock \bibinfo{journal}{{\em Data Mining and Knowledge Discovery\/}}
  \bibinfo{volume}{28}, \bibinfo{number}{5-6} (\bibinfo{year}{2014}),
  \bibinfo{pages}{1158--1188}.
\newblock


\bibitem[\protect\citeauthoryear{Hardt, Price, and Srebro}{Hardt
  et~al\mbox{.}}{2016}]%
        {hardt2016equality}
\bibfield{author}{\bibinfo{person}{Moritz Hardt}, \bibinfo{person}{Eric Price},
  {and} \bibinfo{person}{Nati Srebro}.} \bibinfo{year}{2016}\natexlab{}.
\newblock \showarticletitle{Equality of opportunity in supervised learning}. In
  \bibinfo{booktitle}{{\em Proc. of {NIPS}}}. \bibinfo{publisher}{Curran
  Associates, Inc.}, \bibinfo{pages}{3315--3323}.
\newblock


\bibitem[\protect\citeauthoryear{Jabbari, Joseph, Kearns, Morgenstern, and
  Roth}{Jabbari et~al\mbox{.}}{2016}]%
        {jabbari2016fair}
\bibfield{author}{\bibinfo{person}{Shahin Jabbari}, \bibinfo{person}{Matthew
  Joseph}, \bibinfo{person}{Michael Kearns}, \bibinfo{person}{Jamie
  Morgenstern}, {and} \bibinfo{person}{Aaron Roth}.}
  \bibinfo{year}{2016}\natexlab{}.
\newblock \showarticletitle{Fair Learning in Markovian Environments}.
\newblock \bibinfo{journal}{{\em arXiv:1611.03071\/}} (\bibinfo{year}{2016}).
\newblock


\bibitem[\protect\citeauthoryear{J{\"a}rvelin and
  Kek{\"a}l{\"a}inen}{J{\"a}rvelin and Kek{\"a}l{\"a}inen}{2002}]%
        {jarvelin2002cumulated}
\bibfield{author}{\bibinfo{person}{Kalervo J{\"a}rvelin} {and}
  \bibinfo{person}{Jaana Kek{\"a}l{\"a}inen}.} \bibinfo{year}{2002}\natexlab{}.
\newblock \showarticletitle{Cumulated gain-based evaluation of IR techniques}.
\newblock \bibinfo{journal}{{\em ACM Transactions on Information Systems\/}}
  \bibinfo{volume}{20}, \bibinfo{number}{4} (\bibinfo{year}{2002}),
  \bibinfo{pages}{422--446}.
\newblock


\bibitem[\protect\citeauthoryear{Kamiran, Calders, and Pechenizkiy}{Kamiran
  et~al\mbox{.}}{2010}]%
        {CaldersICDM}
\bibfield{author}{\bibinfo{person}{Faisal Kamiran}, \bibinfo{person}{Toon
  Calders}, {and} \bibinfo{person}{Mykola Pechenizkiy}.}
  \bibinfo{year}{2010}\natexlab{}.
\newblock \showarticletitle{Discrimination aware decision tree learning}. In
  \bibinfo{booktitle}{{\em Proc. of {ICDM}}}. \bibinfo{publisher}{IEEE CS},
  \bibinfo{pages}{869--874}.
\newblock


\bibitem[\protect\citeauthoryear{Kamishima, Akaho, Asoh, and Sakuma}{Kamishima
  et~al\mbox{.}}{2012}]%
        {Kamishima2012}
\bibfield{author}{\bibinfo{person}{Toshihiro Kamishima},
  \bibinfo{person}{Shotaro Akaho}, \bibinfo{person}{Hideki Asoh}, {and}
  \bibinfo{person}{Jun Sakuma}.} \bibinfo{year}{2012}\natexlab{}.
\newblock \showarticletitle{Fairness-aware classifier with prejudice remover
  regularizer}.
\newblock In \bibinfo{booktitle}{{\em Machine Learning and Knowledge Discovery
  in Databases}}. \bibinfo{publisher}{Springer}, \bibinfo{pages}{35--50}.
\newblock


\bibitem[\protect\citeauthoryear{Kleinberg, Mullainathan, and
  Raghavan}{Kleinberg et~al\mbox{.}}{2016}]%
        {kleinberg2016inherent}
\bibfield{author}{\bibinfo{person}{Jon Kleinberg}, \bibinfo{person}{Sendhil
  Mullainathan}, {and} \bibinfo{person}{Manish Raghavan}.}
  \bibinfo{year}{2016}\natexlab{}.
\newblock \showarticletitle{Inherent trade-offs in the fair determination of
  risk scores}.
\newblock \bibinfo{journal}{{\em arXiv:1609.05807\/}} (\bibinfo{year}{2016}).
\newblock


\bibitem[\protect\citeauthoryear{Kulshrestha, Zafar, Eslami, Ghosh, Messias,
  and Gummadi}{Kulshrestha et~al\mbox{.}}{2017}]%
        {kulshrestha_2017_quantifying}
\bibfield{author}{\bibinfo{person}{Juhi Kulshrestha},
  \bibinfo{person}{Muhammad~B. Zafar}, \bibinfo{person}{Motahhare Eslami},
  \bibinfo{person}{Saptarshi Ghosh}, \bibinfo{person}{Johnnatan Messias}, {and}
  \bibinfo{person}{Krishna~P. Gummadi}.} \bibinfo{year}{2017}\natexlab{}.
\newblock \showarticletitle{Quantifying search bias: Investigating sources of
  bias for political searches in social media}. In \bibinfo{booktitle}{{\em
  Proc. of {CSCW}}}.
\newblock


\bibitem[\protect\citeauthoryear{Kunaver and Po{\v{z}}rl}{Kunaver and
  Po{\v{z}}rl}{2017}]%
        {kunaver2017diversity}
\bibfield{author}{\bibinfo{person}{Matev{\v{z}} Kunaver} {and}
  \bibinfo{person}{Toma{\v{z}} Po{\v{z}}rl}.} \bibinfo{year}{2017}\natexlab{}.
\newblock \showarticletitle{Diversity in recommender systems--A survey}.
\newblock \bibinfo{journal}{{\em Knowledge-Based Systems\/}}
  \bibinfo{volume}{123} (\bibinfo{year}{2017}), \bibinfo{pages}{154--162}.
\newblock


\bibitem[\protect\citeauthoryear{Lerner}{Lerner}{2003}]%
        {lerner2003group}
\bibfield{author}{\bibinfo{person}{N{\=a}t{\=a}n Lerner}.}
  \bibinfo{year}{2003}\natexlab{}.
\newblock \bibinfo{booktitle}{{\em Group rights and discrimination in
  international law}}. Vol.~\bibinfo{volume}{77}.
\newblock \bibinfo{publisher}{Martinus Nijhoff Publishers}.
\newblock


\bibitem[\protect\citeauthoryear{Lichman}{Lichman}{2013}]%
        {lichman_2013_uci}
\bibfield{author}{\bibinfo{person}{M. Lichman}.}
  \bibinfo{year}{2013}\natexlab{}.
\newblock \bibinfo{title}{{UCI} Machine Learning Repository}.
\newblock   (\bibinfo{year}{2013}).
\newblock


\bibitem[\protect\citeauthoryear{O'Neil}{O'Neil}{2016}]%
        {oneil2016weapons}
\bibfield{author}{\bibinfo{person}{Cathy O'Neil}.}
  \bibinfo{year}{2016}\natexlab{}.
\newblock \bibinfo{booktitle}{{\em Weapons of math destruction: How big data
  increases inequality and threatens democracy}}.
\newblock \bibinfo{publisher}{Crown Publishing Group}.
\newblock


\bibitem[\protect\citeauthoryear{Pedreschi, Ruggieri, and Turini}{Pedreschi
  et~al\mbox{.}}{2009a}]%
        {pedreschi2009integrating}
\bibfield{author}{\bibinfo{person}{Dino Pedreschi}, \bibinfo{person}{Salvatore
  Ruggieri}, {and} \bibinfo{person}{Franco Turini}.}
  \bibinfo{year}{2009}\natexlab{a}.
\newblock \showarticletitle{Integrating induction and deduction for finding
  evidence of discrimination}. In \bibinfo{booktitle}{{\em Proc. of {AI} and
  {Law}}}. \bibinfo{publisher}{ACM Press}, \bibinfo{pages}{157--166}.
\newblock


\bibitem[\protect\citeauthoryear{Pedreschi, Ruggieri, and Turini}{Pedreschi
  et~al\mbox{.}}{2009b}]%
        {pederruggi2009}
\bibfield{author}{\bibinfo{person}{Dino Pedreschi}, \bibinfo{person}{Salvatore
  Ruggieri}, {and} \bibinfo{person}{Franco Turini}.}
  \bibinfo{year}{2009}\natexlab{b}.
\newblock \showarticletitle{Measuring Discrimination in Socially-Sensitive
  Decision Records}. In \bibinfo{booktitle}{{\em Proc. of {SDM}}}.
  \bibinfo{publisher}{SIAM}, \bibinfo{pages}{581--592}.
\newblock


\bibitem[\protect\citeauthoryear{Pedreshi, Ruggieri, and Turini}{Pedreshi
  et~al\mbox{.}}{2008}]%
        {peder2008}
\bibfield{author}{\bibinfo{person}{Dino Pedreshi}, \bibinfo{person}{Salvatore
  Ruggieri}, {and} \bibinfo{person}{Franco Turini}.}
  \bibinfo{year}{2008}\natexlab{}.
\newblock \showarticletitle{Discrimination-aware data mining}. In
  \bibinfo{booktitle}{{\em Proc. of {KDD}}}. \bibinfo{publisher}{ACM Press},
  \bibinfo{pages}{560--568}.
\newblock


\bibitem[\protect\citeauthoryear{Sakai and Song}{Sakai and Song}{2011}]%
        {sakai2011evaluating}
\bibfield{author}{\bibinfo{person}{Tetsuya Sakai} {and} \bibinfo{person}{Ruihua
  Song}.} \bibinfo{year}{2011}\natexlab{}.
\newblock \showarticletitle{Evaluating diversified search results using
  per-intent graded relevance}. In \bibinfo{booktitle}{{\em Proc. of {SIGIR}}}.
  \bibinfo{publisher}{ACM Press}, \bibinfo{pages}{1043--1052}.
\newblock


\bibitem[\protect\citeauthoryear{Sowell}{Sowell}{2005}]%
        {sowell2005affirmative}
\bibfield{author}{\bibinfo{person}{Thomas Sowell}.}
  \bibinfo{year}{2005}\natexlab{}.
\newblock \bibinfo{booktitle}{{\em Affirmative action around the world: An
  empirical analysis}}.
\newblock \bibinfo{publisher}{Yale University Press}.
\newblock


\bibitem[\protect\citeauthoryear{{The College Board}}{{The College
  Board}}{2014}]%
        {sat_2014}
\bibfield{author}{\bibinfo{person}{{The College Board}}.}
  \bibinfo{year}{2014}\natexlab{}.
\newblock \bibinfo{title}{{SAT} Percentile Ranks}.
\newblock   (\bibinfo{year}{2014}).
\newblock


\bibitem[\protect\citeauthoryear{Verge}{Verge}{2010}]%
        {verge2010gendering}
\bibfield{author}{\bibinfo{person}{T{\`a}nia Verge}.}
  \bibinfo{year}{2010}\natexlab{}.
\newblock \showarticletitle{Gendering representation in Spain: Opportunities
  and limits of gender quotas}.
\newblock \bibinfo{journal}{{\em J. of Women, Politics \& Policy\/}}
  \bibinfo{volume}{31}, \bibinfo{number}{2} (\bibinfo{year}{2010}),
  \bibinfo{pages}{166--190}.
\newblock


\bibitem[\protect\citeauthoryear{Yang and Stoyanovich}{Yang and
  Stoyanovich}{2016}]%
        {yang2016measuring}
\bibfield{author}{\bibinfo{person}{Ke Yang} {and} \bibinfo{person}{Julia
  Stoyanovich}.} \bibinfo{year}{2016}\natexlab{}.
\newblock \showarticletitle{Measuring Fairness in Ranked Outputs}. In
  \bibinfo{booktitle}{{\em Proc. of {FATML}}}.
\newblock


\bibitem[\protect\citeauthoryear{Zafar, Valera, Rodriguez, and Gummadi}{Zafar
  et~al\mbox{.}}{2015}]%
        {zafar2015}
\bibfield{author}{\bibinfo{person}{Muhammad~Bilal Zafar},
  \bibinfo{person}{Isabel Valera}, \bibinfo{person}{Manuel~Gomez Rodriguez},
  {and} \bibinfo{person}{Krishna~P Gummadi}.} \bibinfo{year}{2015}\natexlab{}.
\newblock \showarticletitle{Fairness constraints: A mechanism for fair
  classification}.
\newblock \bibinfo{journal}{{\em arXiv:1507.05259\/}} (\bibinfo{year}{2015}).
\newblock


\bibitem[\protect\citeauthoryear{Zemel, Wu, Swersky, Pitassi, and Dwork}{Zemel
  et~al\mbox{.}}{2013}]%
        {Zemel2013}
\bibfield{author}{\bibinfo{person}{Rich Zemel}, \bibinfo{person}{Yu Wu},
  \bibinfo{person}{Kevin Swersky}, \bibinfo{person}{Toni Pitassi}, {and}
  \bibinfo{person}{Cynthia Dwork}.} \bibinfo{year}{2013}\natexlab{}.
\newblock \showarticletitle{Learning fair representations}. In
  \bibinfo{booktitle}{{\em Proc. of {ICML}}}. \bibinfo{pages}{325--333}.
\newblock


\bibitem[\protect\citeauthoryear{{\v{Z}}liobaite, Kamiran, and
  Calders}{{\v{Z}}liobaite et~al\mbox{.}}{2011}]%
        {vzliobaite2011handling}
\bibfield{author}{\bibinfo{person}{Indre {\v{Z}}liobaite},
  \bibinfo{person}{Faisal Kamiran}, {and} \bibinfo{person}{Toon Calders}.}
  \bibinfo{year}{2011}\natexlab{}.
\newblock \showarticletitle{Handling conditional discrimination}. In
  \bibinfo{booktitle}{{\em Proc. of {ICDM}}}. \bibinfo{publisher}{IEEE CS},
  \bibinfo{pages}{992--1001}.
\newblock


\end{thebibliography}




\end{document}